\documentclass[reqno,a4,11pt]{amsart}
\usepackage{amsmath,amssymb,tabularx,setspace}
\usepackage{color,graphics,lscape,rotating}
\usepackage{pdflscape}
\usepackage{adjustbox}
\usepackage{multirow}
\usepackage{tabularx}
\newtheorem{theorem}{Theorem}

\newtheorem{definition}{Definition}
\newtheorem{remark}{Remark}
\newtheorem{corollary}{Corollary}
\makeatletter
\renewcommand\section{\@startsection {section}{1}{\z@}%
                                   {-3.5ex \@plus -1ex \@minus -.2ex}%
                                   {2.3ex \@plus.2ex}%
                                   {\normalfont\large\bfseries}}
\makeatother

\newtheorem{Corollary}{Corollary}

\newtheorem{Theorem}{Theorem}

\newtheorem{Example}{Example}

\begin{document}
\doublespace
\vspace{-0.3in}
\title[]{ Jackknife empirical likelihood ratio test for log symmetric distribution using probability weighted moments }
\author[]
{   A\lowercase{njana}  S\lowercase{\textsuperscript{a} and} S\lowercase{udheesh} K. K\lowercase{attumannil,\textsuperscript{b}  }
\\
 \lowercase{\textsuperscript{a}}U\lowercase{niversity of} H\lowercase{yderabad } H\lowercase{yderabad}, I\lowercase{ndia},\\ \lowercase{\textsuperscript{b}}I\lowercase{ndian} S\lowercase{tatistical} I\lowercase{nstitute},
  C\lowercase{hennai}, I\lowercase{ndia}.}
\maketitle
\vspace{-0.2in}
\maketitle

\begin{abstract} Log symmetric distributions are useful in modeling data which show high skewness and have found applications in various fields. Using a recent characterization for log symmetric distributions, we propose a goodness of fit test for testing log symmetry. The asymptotic distributions of the test statistics under both null and alternate distributions are obtained.  As the normal-based test is difficult to implement, we also propose a jackknife empirical likelihood (JEL) ratio test for testing log symmetry. We conduct a Monte Carlo Simulation to evaluate the performance of the JEL ratio test. Finally, we illustrated our methodology using different data sets.    

\end{abstract}

\section{Introduction}\label{sec1}
The family of log-symmetric distributions includes several specific types of distributions commonly used to model continuous,  positive, and asymmetric data. It is flexible enough to handle bimodal data, as well as distributions with light or heavy tails. Log-symmetric distributions are especially useful for dealing with data that show significant skewness.  This family includes several well-known distributions, such as the log-logistic, log-Laplace, log-Cauchy, log-power-exponential, log-student-t,  and Birnbaum-Saunders distributions. Some log-symmetric distributions have heavier tails than the log-normal distribution, making them robust for estimating model parameters even in the presence of extreme or outlying observations.  This class of distribution is extensively studied by Vanegas and Paula (2016), Ferrari and Fumes (2017) and Ahmadi and Balakrishnan (2024).

Let $X$ be a continuous positive random variable with distribution function $F$. Then the distribution is said to possess log-symmetry about $\theta$ if   $\frac{X}{\theta}$ and $\frac{\theta}{X}$ have same distribution for some $\theta>0$.  This property is referred to as log symmetry as it is equivalent to ordinary symmetry of the distribution of $\log(X)$ about $\log(\theta)$ (Marshall and Olkin  (2007), Seshadri  (1965)).  Thus we have the following definition.
\begin{definition}\label{def1}
  Let $X$ be a continuous positive random variable with the property that $X/\theta$ and $\theta/X$ are identically distributed then $X$ has a log symmetric distribution about the point $\theta$.
  \end{definition}

 Several studies in the literature have explored the class of log-symmetric distributions and their properties. However, to the best of our knowledge, no testing procedure exists to test log-symmetry. This highlights the need to develop a test for log-symmetry. If the data is found to exhibit the log-symmetric property, subsequent tests can be done to determine which specific distribution within this class best describes the data. This motivates us to develop a test for log symmetric distributions. This paper aims to propose a nonparametric test for testing log-symmetry. We used a recent characterization of log-symmetric distributions established by  Ahmadi and Balakrishnan (2024) to develop the test.

 Probability weighted moments (PWMs) are less affected by outliers than traditional moments, making them more reliable for handling skewed and heavy-tailed distributions. Hence, in the case of log-symmetric distributions with heavy tails, PWMs can be used to estimate the parameters of the distribution.
 This motivates us to use PWMs in developing the proposed test.  For some recent developments on PWMs, see Vexler et al. (2017)  and Deepesh et al. (2021) and  Sudheesh et al. (2024) and the references therein.

 The rest of the paper is organized as follows. In Section 2, we develop a test for testing log symmetry and  derive the asymptotic distributions of the test statistics under both null and alternate distributions are obtained. Also, we propose a jackknife empirical likelihood (JEL) ratio test for testing log symmetry.  Monte Carlo Simulations are carried out to evaluate the performance of the JEL ratio test and the results are presented in Section 3. Numerical illustrations of the proposed test using different data sets are conducted and are presented in Section 4. Concluding remarks are given in Section 5.

\section{Test statistic}\label{sec2}
Let $X$ be a non-negative continuous random variable having distribution function $F(x)$. Here, we develop a nonparametric test for testing log-symmetry. As mentioned in the introduction, we propose a test based on a recent characterization of log-symmetric distributions by Ahmadi and Balakrishnan (2024).

Hence we consider the following characterization result given in Theorem 3.2 of Ahmadi and Balakrishnan (2024).
\begin{theorem}
 Let $X$ be a positive continuous random variable having distribution function $F$,  and $g$ is a positive real-valued strictly monotonic continuous function defined on the support of $F$. Assume $E(g(X))$ and $E(g(1/X))$ both exist. $X$ has a log-symmetric distribution around 1 if and only if
 $$E(g(X))-E(g(1/X))=0.$$
\end{theorem}Since we are focused on developing the test based on probability weighted moments (PWMs), we first give the definition of PWMs.
\begin{definition}\label{pwm}
    For a random variable $X$ with cumulative distribution function  $F(x)$, the probability weighted moments are defined as:
$$\beta_r=E[X(F(X))^r],\,r>0.$$
\end{definition}
Given Definition \ref{pwm}, we choose $g(x)=xF^{\beta}(x)$ to obtain a characterization of log symmetric distribution using PWMs. Thus we have the following result.
\begin{corollary}\label{cor1}
     Let $X$ be a positive continuous random variable having distribution function $F$. $X$ has a log-symmetric distribution around 1 if and only if
 \begin{equation}\label{charect}
     E\left(XF^{\beta}(X)\right)-E\left(\frac{1}{X}F^{\beta}(\frac{1}{X})\right)=0.
 \end{equation}
\end{corollary}
\subsection{Test based on U-statistics}
 Let  $\mathcal{F}$ be the class of distribution having the property of log symmetry given in Definition 1. Without loss of generality, we assume $\theta=1$. The proposed test can be modified for the general $\theta$. See Remark 1 for more details.

 Consider a   random sample $X_{1}, ...,X_{n}$  from $ {F}$. Based on this sample,  we are interested in testing the null hypothesis
\begin{equation*}
    \mathcal{H}_0 : F\in \mathcal{F}
\end{equation*}
against the alternative hypothesis
\begin{equation*}
    \mathcal{H}_1 : F \notin \mathcal{F}.
\end{equation*}
To test the above hypothesis, we first define a departure measure $\Delta(F)$ which discriminates between the null and alternative hypothesis. In view of Corollary \ref{cor1} we define the departure measure as
\begin{eqnarray}\label{delta}
 \Delta(F)&=&\left(E(XF^{\beta}(X))-E(X^{-1}F^{\beta}(\frac{1}{X}))\right).
\end{eqnarray}Since we are interested in developing a test  based on U-statistics, we express the departure measure defined in (\ref{delta}) as an expectation of functions of random variables.  In this process, note that the distribution function  of $\max(X_1,\ldots,X_{\beta})$ is given by $F^{\beta}(x)$. Let $\bar{F}(x)=1-F(x)$ be the survival function of $X$ at the point $x$. Thus the survival function of $\min(X_1,\ldots,X_{\beta})$ becomes $\bar{F}^{\beta}(x)$.
Now, consider
\begin{eqnarray}
 \Delta(F)&=& E\left(XF^{\beta}(X)-X^{-1}F^{\beta}(\frac{1}{X})\right)\nonumber\\ &=&\frac{1}{(\beta+1)} \int_{0}^{\infty} (\beta+1)\Big(xF^{\beta}(x)-\frac{1}{x}F^{\beta}(\frac{1}{x})\Big) dF(x)\nonumber
 \\ &=& \frac{1}{(\beta+1)} \int_{0}^{\infty} (\beta+1)\Big(xF^{\beta}(x)-\frac{1}{x}(1-\bar{F}(\frac{1}{x}))^{\beta}\Big) dF(x)\nonumber\\ &=&\frac{1}{(\beta+1)}  \int_{0}^{\infty} (\beta+1)\Big(xF^{\beta}(x)-\frac{1}{x}(1-F(x))^{\beta}\Big) dF(x)\nonumber\\
  &=& \frac{1}{(\beta+1)} E\left(\max(X_1,\ldots,X_{\beta+1})-\frac{1}{\min(X_1,\ldots,X_{\beta+1})}\right),\label{deltasimp}
\end{eqnarray}
where  $X_1, X_2,\ldots,X_{\beta+1}$ are independently distributed  random variables from ${F}$. One can easily verify that $\Delta(F)=0$ under $H_0$ by considering different choices of log symmetric distributions.

We find the test statistic using the theory of U-statistics. Consider the symmetric kernel
\begin{small}
    \begin{equation}
        \label{kernel}
        h(X_{1},\cdots,X_{\beta+1})=\frac{1}{(\beta+1)} \left(\max(X_{1},\cdots,X_{\beta+1})-\frac{1}{\min(X_{1},\cdots,X_{\beta+1})}\right)
    \end{equation}
\end{small}so that $E(h(X_{1},\cdots,X_{\beta+1}))=\Delta(F)$. Hence the test statistic is given by
\begin{equation}\label{test}
   \widehat{\Delta} =\binom{n}{\beta+1}^{-1}\sum\limits_{C_{\beta+1,n}}h{(X_{i_1},X_{i_2},\cdots,X_{i_{\beta+1}})},
\end{equation}
where the summation is over the set  $C_{\beta+1,n}$ of all combinations of $(\beta+1)$ distinct elements $\lbrace i_1,i_2,\cdots,i_{\beta+1}\rbrace$ chosen from $\lbrace 1,2,\cdots,n\rbrace$.
Now we simplify (\ref{test})  in terms of order statistics. Let $X_{(i)}$ be the $i$th order statistics based on a random sample  $X_{1}, ...,X_{n}$  from ${F}$. Then, we have the following expressions:
\begin{equation*}
   \sum_{i=1}^{n}\sum_{j=1; j<i}^{n}\max(X_i,X_j) =\sum_{i=1}^{n}(i-1)X_{(i)}
\end{equation*}
and
\begin{equation*}
  \sum_{i=1}^{n}\sum_{j=1; j<i}^{n}\sum_{k=1; k<j}^{n}\max(X_i,X_j,X_k) =\frac{1}{2}\sum_{i=1}^{n}(i-1)(i-2)X_{(i)}.
\end{equation*}
In general, we have
\begin{equation}\label{max}
    \sum\limits_{C_{\beta+1,n}} \max(X_{i_1},\ldots,X_{i_{\beta+1}})=\sum_{i=1}^{n}\dbinom{i-1}{\beta}X_{(i)}.
\end{equation}
Also, we have
\begin{equation*}
   \sum_{i=1}^{n}\sum_{j=1; j<i}^{n}\min(X_i,X_j) =\sum_{i=1}^{n}(n-i)X_{(i)}
\end{equation*}
and
\begin{eqnarray*}
    \sum_{i=1}^{n}\sum_{j=1; j<i}^{n}\sum_{k=1; k<j}^{n}\min(X_i,X_j,X_k) &=&\sum_{i=1}^{n}\frac{(n-i-1)(n-i)}{2}X_{(i)}\\&=&\sum_{i=1}^{n}\dbinom{n-i}{2}X_{(i)}.
\end{eqnarray*}
With some algebraic manipulation, we obtain
\begin{equation}\label{min}
    \sum\limits_{C_{\beta+1,n}} \frac{1}{\min(X_{i_1},\ldots,X_{i_{\beta+1}})}=\sum_{i=1}^{n}\binom{n-i}{\beta}\frac{1}{X_{(i)}}.
\end{equation}
Using (\ref{max}) and (\ref{min}) we can rewrite (\ref{test}) as
\begin{equation}\label{testsimple}
     \widehat{\Delta} =\binom{n}{\beta+1}^{-1}\frac{1}{(\beta+1)}\sum_{i=1}^{n}\left(\dbinom{i-1}{\beta}X_{(i)}-\binom{n-i}{\beta}\frac{1}{X_{(i)})}\right).
\end{equation}
The test procedure to reject the null hypothesis $H_0$ against the alternative hypothesis $H_1$ for large values of $\widehat{\Delta}$.
\begin{remark}The proposed test can be generalized to test the log symmetry about a point $\theta>0$. Given Definition 1 and Corollary 1, we can  modify the departure measure in (\ref{deltasimp}) as
\begin{eqnarray*}
 \Delta^*(F)=\frac{1}{(\beta+1)}E\left(\frac{\max(X_1,\ldots,X_{\beta+1})}{\theta}-\frac{\theta}{\min(X_1,\ldots,X_{\beta+1})}\right).
\end{eqnarray*}In this case, the test statistics becomes
\begin{equation}\label{testsimplenew}
     \widehat{\Delta}^* =\binom{n}{\beta+1}^{-1}\frac{1}{(\beta+1)}\sum_{i=1}^{n}\left(\dbinom{i-1}{\beta}\frac{X_{(i)}}{\theta}-\binom{n-i}{\beta}\frac{\theta}{X_{(i)})}\right).
\end{equation} For example, if $X$ has log-normal distribution with parameters $\mu$ and $\sigma$, then test statistics in (\ref{testsimplenew}) becomes
\begin{equation}\label{testsimplenewlog}
     \widehat{\Delta}^* =\binom{n}{\beta+1}^{-1}\frac{1}{(\beta+1)}\sum_{i=1}^{n}\left(\dbinom{i-1}{\beta}\frac{X_{(i)}}{\exp(\mu)}-\binom{n-i}{\beta}\frac{\exp(\mu)}{X_{(i)})}\right).
\end{equation}
\end{remark}
    We find a critical region of the proposed test based on the asymptotic distribution of the $\widehat{\Delta}$.  As the test statistic is a U-statistic $ \widehat{\Delta}$ ia consistent estimator of $ {\Delta}(F)$.  Next, we obtain the asymptotic distribution of $ \widehat{\Delta}$.
\begin{theorem}
    As $n\to \infty$, $\sqrt{n}(\widehat{\Delta}-\Delta(F))$ converges in distribution to a normal random variable with mean zero and variance $\sigma^2$, where $\sigma^2$ is given by
\begin{eqnarray}
    \label{asympvar}
    \sigma^2&=& Var\Bigg(X F^{\beta}(X)+\int_{X}^{\infty}y(\beta+1) F^{\beta}(y)dF(y)-\frac{1}{X} \bar{F}^{\beta}(X)\nonumber\\&&\qquad\qquad\qquad\qquad-\int_{0}^{X}\frac{1}{y}(\beta+1) \bar{F}^{\beta}(y)dF(y)\Bigg).
\end{eqnarray}
\end{theorem}

\begin{proof}
Using the central limit theorem for U-statistics, as $n \to \infty$, $\sqrt{n}(\widehat{\Delta}-\Delta(F))$ converges in distribution to a normal random variable with mean zero and variance $(\beta+1)^2\sigma^2$, where $\sigma^2$ is the asymptotic variance of $\widehat{\Delta}$ and is given by (Lee, 2019)
\begin{equation}\label{varexp}
    \sigma^2= Var\big(E(h(X_1, X_2,\dots, X_{\beta+1})|X_1)\big).
\end{equation}
Denote $Z_1=\max(X_{2},\cdots,X_{\beta+1})$, $Z_2=\min(X_{2},\cdots,X_{\beta+1})$ and $I$ denotes the indicator function. Consider
 \begin{eqnarray}\label{varexpfinal}
    && \hskip-0.4inE(h(X_1, X_2, \ldots, X_{\beta+1})|X_1=x)\nonumber \\&=&\frac{1}{(\beta+1)}E\Bigg(\max(x,X_{2},\cdots,X_{\beta+1})-\frac{1}{\min(x,X_{2},\cdots,X_{\beta+1})}\Bigg)\nonumber\\
     &=&\frac{1}{(\beta+1)}E\Bigg(x I(x>Z_1)+ Z_1I(Z_1>x)-\frac{1}{x}I(x<Z_2)- \frac{1}{Z_2}I(Z_2<x)\Bigg)\nonumber\\
     &=&\frac{1}{(\beta+1)}\Bigg(x F^{\beta}(x)+\int_{x}^{\infty}y(\beta+1) F^{\beta}(y)dF(y)-\frac{1}{x} \bar{F}^{\beta}(x) \nonumber\\&& \qquad\qquad\qquad\qquad- \int_{0}^{x}\frac{1}{y}(\beta+1) \bar{F}^{\beta}(y)dF(y)\Bigg).
 \end{eqnarray}
Substituting (\ref{varexpfinal}) in (\ref{varexp}), we have the asymptotic variance specified in the equation (\ref{asympvar}).  Hence the proof of the theorem.

\end{proof}

Next, we obtain the asymptotic null distribution of the test statistic. Note that under $H_0$, $\Delta(F)=0$. Hence we have the following result.
\begin{corollary}
    Under $H_0$, as $n\to \infty$, $\sqrt{n}\widehat{\Delta}$ converges in distribution to a normal random variable with mean zero and variance $\sigma^2_0$, where $\sigma^2_0$ is the value of the asymptotic variance  $\sigma^2$ evaluate under $H_0$. 
    \end{corollary}
Using Corollary 1, we can obtain a test based on normal approximation and we reject the null hypothesis $H_0$ against the alternative $H_1$ if
\begin{equation*}
    \dfrac{\sqrt{n}|\widehat{\Delta}|}{\hat{\sigma}_0}\geq Z_{\alpha/2},
\end{equation*}
where $\hat{\sigma}_0$ is a consistent estimator of $\sigma_0$ and $Z_\alpha$ is the upper $\alpha$-percentile point of a standard normal distribution. Implementing the test based on normal approximation is not simple as finding a consistent estimator of $\sigma^2_0$ is very difficult. This motivates us to develop an empirical likelihood-based test that is distribution-free. Next, we discuss JEL-based tests for testing log-symmetry.



\section{Jackknife empirical likelihood ratio test}
Thomas and Grunkemier (1975) introduced the concept of empirical likelihood to obtain a confidence interval for survival probabilities under right censoring. Owen's seminal papers (Owen (1988), Owen (1990)) extended the empirical likelihood to a general methodology. In empirical likelihood, we need to maximize non-parametric likelihood supported by the data subject to some constraints. If the constraints are not linear (U-statistics with a degree greater than 2), computational difficulties are encountered while evaluating the likelihood. To overcome this problem, Jing et al.  (2009) introduced the jackknife empirical likelihood (JEL) method to find a confidence interval for the desired parametric function. They illustrated the JEL methodology using one sample and two sample U-statistics. This approach is widely accepted among researchers as it combines the effectiveness of the likelihood approach and the non-parametric nature of the jackknife resampling technique. The empirical likelihood and JEL are particularly useful when the distribution is not completely specified under the null; See  Sudheesh et al.  (2022) Deemat et al. (2024), Saparya and Sudheesh (2024) and Jiang and Zhao (2024).  As we are interested in developing a test for a class property (log symmetry), JEL can be effectively used.

\par Next, we discuss how to implement the JEL ratio test for testing the log symmetry of a distribution.  First, we find jackknife pseudo-values for developing the test. Denote $\widehat{\Delta}_{k}$, $k=1,2,...,n$ are  the value of the test statistics $\widehat{\Delta}$ obtained using $(n-1)$ observations  $X_1$, $X_2$,..., $X_{k-1}$,$ X_{k+1}$,...,$X_n$; $k=1,2,...,n$.    Then jackknife pseudo-values for $\widehat{\Delta}$ are given by
\begin{equation}\label{jsv}
\widehat{V}_{k}= n \widehat{\Delta}-(n-1)\widehat{\Delta}_{k}; \qquad k=1,2,\cdots,n.
\end{equation}
We use the pseudo-values $\widehat{V}_{k}$ in the empirical likelihood to develop the test.

	Let $p=(p_1,\ldots,p_n)$ be a probability vector assigned to $\widehat{V}_{k}$, $k=1,2,\cdots,n$.  It is well-known that $\prod_{i=1}^{n}p_i$ subject to $\sum_{i=1}^{n}p_i=1$ attain its maximum value $n^{-n}$ at $p_i=1/n$.
	Hence the jackknife empirical likelihood ratio for testing logsymmetry based on the departure measure $\Delta(F)$ is defined as
	$$R(\Delta)=\max\Big\{\prod_{i=1}^{n} np_i,\,\,\sum_{i=1}^{n}p_i=1,\,\,\sum_{i=1}^{n}p_i\nu_i=0\Big\},$$
	where
	$$p_i=\frac{1}{n}\frac{1}{1+\lambda \nu_i}$$
	and  $\lambda$ satisfies $$\frac{1}{n}\sum_{i=1}^{n}\frac{\nu_i}{1+\lambda \nu_i}=0.$$
	Hence the jackknife empirical log-likelihood ratio is given by $$\log R(\Delta)=-\sum\log(1+\lambda \nu_i).$$
	Next theorem explains the limiting distribution of $\log R(\Delta)$ which can be used to construct the JEL ratio test for log symmetry based on $\Delta$.  Using Theorem 1 of Jing et al. (2009) we have the following result as an analog of Wilk's theorem.

\begin{theorem}
  Let  $E\left(h^2(X_{1}, X_{2},..., X_{\beta+1})\right)<\infty$ and $\sigma^2=Var(g(X))$ $>0,$  where $g(X)=E\left(h(X_1, X_{2},..., X_{\beta+1})|X_1=X\right)$. Under $H_0$, as $n\rightarrow\infty$, $-2\log R(\Delta)$ converges in  distribution to a $\chi^{2}$ random variable with one degree of freedom.
\end{theorem}

Using Theorem 3, we develop a JEL ratio test for testing log symmetry.  We reject the null hypothesis $H_0$ against $H_1$ at a  significance level $\alpha$ if $$-2\log R(\Delta)> \chi^2_{1,1-\alpha},$$
where $\chi^2_{1,1-\alpha}$ is a $(1-\alpha)$ percentile point of a $\chi^2$ random variable with one degree of freedom.
\section{Simulation study}\label{sec}
We conduct a Monte Carlo simulation study to assess the performance of the JEL ratio test for log symmetry. We use R software to conduct the simulation. The simulation is repeated ten thousand times considering different sample sizes ($n=25,50,75,100,200$) taking $s=1,2,3$.  In the simulation, the empirical likelihood is evaluated using the `emplik' package available in R.

First, we find the empirical type-I error of the test. For this purpose, we generated a sample from
log-normal, log-logistic, log-Laplace, log-Cauchy
and Birnbaum-Saunders distributions. The log-symmetric distributions are applications in various fields as they are useful in modeling data exhibiting a high degree of skewness. The empirical type-I error is presented in Tables 1 and 2. In Table 1 we reported the empirical type-I error obtained for lognormal distribution with various parameter settings. In this case, we use the test statistics in (\ref{testsimplenewlog}) to find the jackknife pseudo values. In Table 2, we tabulated the empirical type-I error obtained for standard log-logistic, log-Laplace, log-Cauchy  and Birnbaum-Saunders distributions. From Tables 1 and 2, we observe that the empirical type-I error converges to the chosen significance level as the sample size increases. We noted that for the small values of $n$, the proposed test has a type I error that is a little higher than the significance level.

\begin{table}[h]
\caption{Empirical type-I error of the JEL ratio test at significance level $0.05$: lognormal distributions with parameters $\mu$ and $\sigma$. }
\begin{tabular}{ccccccccccccccc}\hline
&$n$&$(\mu,\sigma)=(0,1)$ &$(\mu,\sigma)=(1,1)$&$(\mu,\sigma)=(2,1)$&$(\mu,\sigma)=(3,1)$\\\hline
\multirow{6}{*}{$\beta=1$}
&25&0.118&0.104&0.110&0.118\\
&50&0.091&0.092&0.087&0.089\\
&75&0.083&0.069&0.094&0.086\\
&100&0.065&0.075&0.090&0.082\\
&200&0.055&0.057&0.068&0.075\\
&500&0.051&0.050&0.054&0.055\\
\hline
\multirow{5}{*}{$\beta=2$}
&25&0.122&0.124&0.123&0.124\\
&50&0.092&0.096&0.101&0.105\\
&75&0.092&0.071&0.082&0.079\\
&100&0.071&0.090&0.076&0.067\\
&200&0.065&0.075&0.063&0.064\\
&500&0.054&0.054&0.055&0.052
\\\hline\multirow{4}{*}{$\beta=3$}
&25&0.144&0.139&0.127&0.143\\
&50&0.108&0.116&0.113&0.114\\
&75&0.104&0.097&0.095&0.098\\
&100&0.076&0.077&0.085&0.077\\
&200&0.067&0.059&0.062&0.065\\
&500&0.054&0.051&0.052&0.054\\
 \hline
\end{tabular}
\end{table}

\begin{table}[h]
\caption{Empirical type-I error of the JEL ratio test at significance level $0.05$: Different log symmetric distributions. }
\begin{tabular}{ccccccccccccccc}\hline
&$n$&log-logistic & log-Laplace&log-Cauchy & Birnbaum-Saunders\\\hline
\multirow{6}{*}{$\beta=1$}
&25&0.118&0.100&0.101&0.116\\
&50&0.090&0.0902&0.089&0.089\\
&75&0.084&0.067&0.093&0.082\\
&100&0.060&0.070&0.091&0.080\\
&200&0.054&0.055&0.065&0.065\\
&500&0.050&0.050&0.052&0.051\\
\hline
\multirow{5}{*}{$\beta=2$}
&25&0.120&0.122&0.123&0.123\\
&50&0.090&0.092&0.100&0.100\\
&75&0.084&0.074&0.080&0.076\\
&100&0.070&0.0740&0.066&0.066\\
&200&0.065&0.065&0.060&0.061\\
&500&0.052&0.051&0.051&0.052\\\hline\multirow{4}{*}{$\beta=3$}
&25&0.140&0.139&0.130&0.140\\
&50&0.110&0.115&0.114&0.111\\
&75&0.91&0.087&0.085&0.088\\
&100&0.071&0.070&0.065&0.067\\
&200&0.061&0.058&0.062&0.065\\
&500&0.051&0.051&0.052&0.050\\
 \hline
\end{tabular}
\end{table}
 Next, we find the empirical power of the JEL  ratio test.  We consider the following distribution to find the
power:
\begin{itemize}
        \item Weibull distribution: $F_{1}(x)=1-e^{(-x/\lambda)^{k}}$, $x>0$, $k,\,\lambda>0$.
         \item Gamma distribution: $F_{2}(x)=\frac{1}{\Gamma(k)} \gamma(k,\frac{x}{\lambda})$, $k,\lambda >0 $ and $\gamma(k,\frac{x}{\lambda})$ is the lower incomplete gamma function.
         \item  Pareto distributions: $F_{3}(x)=(\lambda/x)^{\alpha}$ $x>0$, $\alpha,\, \lambda>0$.
    \item Half-normal distribution (HN ($\sigma$)): $F_{4}(x)= erf(\frac{x}{\sigma\sqrt{2}})$, $\sigma >0$, where \textit{erf} is the error function.
        \end{itemize}
        The empirical power obtained in the simulation study is presented in Table 3. From Table 3, we observe that the JEL ratio test has good power for various alternatives. Also, the power increases as the sample size increases.

        \begin{landscape}
        \begin{table}[h]
\caption{Empirical power of the JEL ratio test for various alternatives at significance level $0.05$ }
\begin{tabular}{ccccccccccccccc}\hline
&$n$ &Gamma (1,0.5) & Gamma (2,1) &Pareto (1,0.5)&Pareto (2,1) &Weibull (1,0.5)&Weibull (1,2)&HN (1) &HN (2)\\\hline
\multirow{5}{*}{$\beta=1$}\\
&25&0.243&0.563&0.811&1.000&1.000&0.247&0.987&0.269\\
&50&0.245&0.660&0.980&1.000&1.000&0.250&1.000&0.353\\
&75&0.375&0.755&0.996&1.000&1.000&0.333&1.000&0.443\\
&100&0.389&0.762&1.000&1.000&1.000&0.374&1.000&0.559\\
&200&0.620&0.869&1.000&1.000&1.000&0.604&1.000&0.770\\
\hline
\multirow{5}{*}{$\beta=2$}\\
&25&0.274&0.550&0.890&1.000&1.000&0.874&0.997&0.329\\
&50&0.316&0.557&0.984&1.000&1.000&0.988&1.000&0.414\\
&75&0.417&0.646&1.000&1.000&1.000&1.000&1.000&0.588\\
&100&0.513&0.677&1.000&1.000&1.000&0.999&1.000&0.680\\
&200&0.717&0.801&1.000&1.000&1.000&1.000&1.000&0.893\\
\hline\multirow{5}{*}{$\beta=3$}\\
&25&0.299&0.482&0.934&   1.000&  1.000&0.904&0.993&0.361\\
& 50&0.371&0.532&0.991&  1.000&  1.000&0.991&1.000&0.526\\
&75&0.463&0.574&1.000&  1.000&  1.000&0.999&1.000&0.635\\
&100&0.538&0.611&1.000&  1.000&  1.000&1.000&1.000&0.768\\
& 200&0.798&0.695&1.000&  1.000&  1.000&1.000&1.000&0.943
\\ \hline
\end{tabular}
\end{table}
 \end{landscape}

\section{Data analysis}\label{sec}
In this Section, we demonstrate the use of our test using two different data sets. First, we consider the tensile strength data given in Sathar  and Jose (2023). The data is given in Table 4.  The data set consists of 50 observations detailing the tensile strength (MPa) of fiber. This data set follows a normal distribution with mean $\mu=3076.88$ and $\sigma=344.362$. The boxplot for the data is Figure 1. The boxplot shows that the data is symmetric. The skewness of the data is obtained as 0.396, which is also an indication that the data is symmetric. For this data, the value of $-2 \log R(\Delta)$ is 1731.703, which is greater than 3.84. Thus we reject the null hypothesis that the data follows a distribution that belongs to the log symmetric class of distributions.

\begin{table}[ht]
    \centering
    \caption{Tensile strength data }
    \begin{tabular}{|c|c|c|c|c|c|c|c|c|}
    \hline
    0.746 & 0.357 & 0.376 & 0.327 & 0.485 & 1.741 & 0.241 & 0.777 & 0.768 \\ \hline
    0.409 & 0.252 & 0.512 & 0.534 & 1.656 & 0.742 & 0.378 & 0.714 & 1.121 \\ \hline
    0.597 & 0.231 & 0.541 & 0.805 & 0.682 & 0.418 & 0.506 & 0.501 & 0.247 \\ \hline
    0.922 & 0.880 & 0.344 & 0.519 & 1.302 & 0.275 & 0.601 & 0.388 & 0.450 \\ \hline
    0.845 & 0.319 & 0.486 & 0.529 & 1.547 & 0.690 & 0.676 & 0.314 & 0.736 \\ \hline
    0.643 & 0.483 & 0.352 & 0.636 & 1.080 &      &      &      &       \\ \hline
    \end{tabular}
  \end{table}

\begin{figure}
    \centering
    \includegraphics[width=0.8\linewidth]{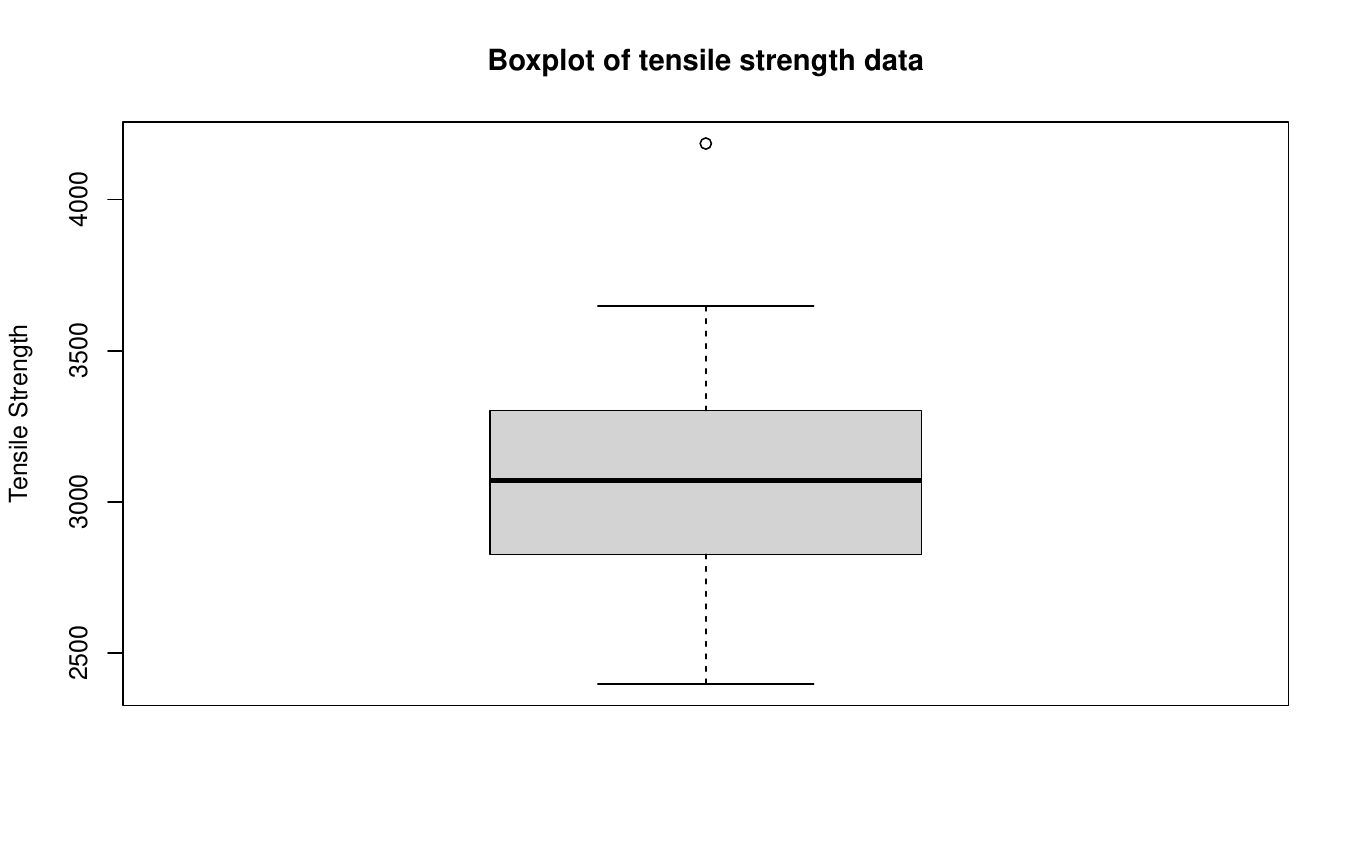}
    \caption{Boxplot of tensile strength data}
    \label{fig:enter-label}
\end{figure}

 Next, we consider the log normal data discussed in Batsidis et. al (2016). The data consists of the gap between two plates measured (in cm)($X$), for each of 50 welded assemblies from the output of a welding process. Data is given in Table 5. For this data, we plotted the Q-Q plot corresponding to the log-normal distribution as shown in Figure 2. For testing log symmetry about point 1, we consider the transformation $Y=\exp((\log(X)-\mu)/\sigma)$ where $\mu$ and $\sigma$ are mean and standard deviation parameters of the log-normal distribution. We find these values using the `EnvStats' package available in R. The value of $-2 \log R(\Delta)$ for the transformed data is calculated to be $0.233$ which is below the threshold of 3.84. We conclude that the data follows log normal distribution, which belongs to the class of log symmetric distribution.

\begin{table}[h]
    \centering
    \caption{The gap between two plates is (measured in cm) }
    \begin{tabular}{|c|c|c|c|c|c|c|c|c|}
    \hline
    0.746 & 0.357 & 0.376 & 0.327 & 0.485 & 1.741 & 0.241 & 0.777 & 0.768 \\ \hline
    0.409 & 0.252 & 0.512 & 0.534 & 1.656 & 0.742 & 0.378 & 0.714 & 1.121 \\ \hline
    0.597 & 0.231 & 0.541 & 0.805 & 0.682 & 0.418 & 0.506 & 0.501 & 0.247 \\ \hline
    0.922 & 0.880 & 0.344 & 0.519 & 1.302 & 0.275 & 0.601 & 0.388 & 0.450 \\ \hline
    0.845 & 0.319 & 0.486 & 0.529 & 1.547 & 0.690 & 0.676 & 0.314 & 0.736 \\ \hline
    0.643 & 0.483 & 0.352 & 0.636 & 1.080 &      &      &      &       \\ \hline
    \end{tabular}
\end{table}
\begin{figure}[h]
    \centering
    \includegraphics[width=0.8\linewidth]{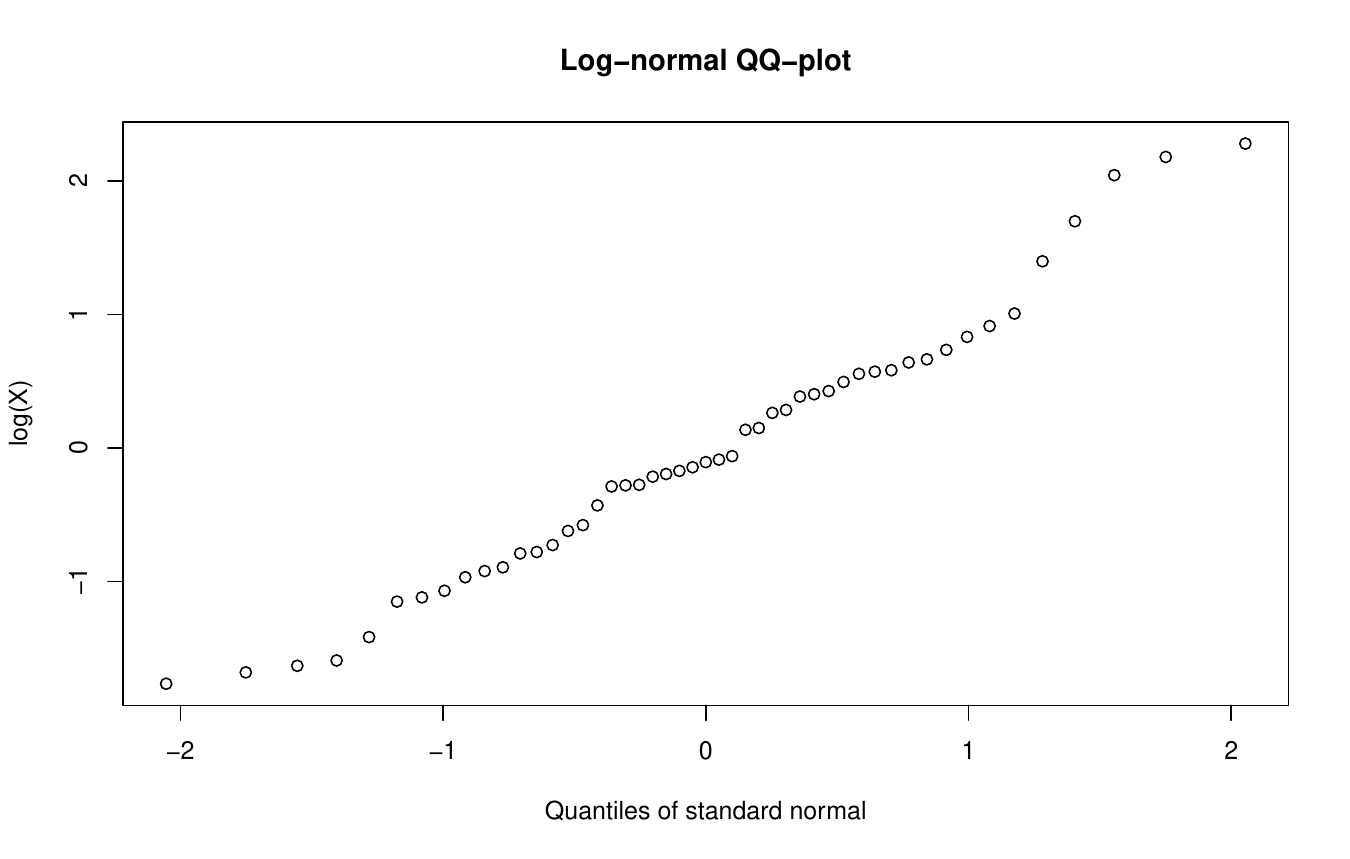}
    \caption{QQ plot of gap between two plates  (measured in cm)}
    \label{fig:enter-label}
\end{figure}
\newpage
\section{Concluding remarks}\label{}
Log symmetric distributions are useful in modeling data that show high skewness and have found applications in various fields. Ahamedi and Balakrishnan (2024) obtained several characterization results for log symmetric distribution. Motivated by their results, we obtained a characterization log of symmetric distributions based on probability-weighted moments. Using the characterization based on probability-weighted moments, we propose a goodness of fit test for testing log symmetry. The asymptotic distributions of the test statistics under both null and alternate distributions are obtained.  As the normal-based test is difficult to implement, we also propose a jackknife empirical likelihood (JEL) ratio test for testing log symmetry. We conduct a Monte Carlo Simulation to evaluate the performance of the JEL ratio test. Simulation results show that the empirical type-I error converges to the chosen significance level as the sample size increases and has good power for different choices of alternatives.  Finally, we illustrated our methodology using various data sets.

The log-symmetric distributions, defined on the positive real line, are commonly used in lifetime data analysis. For example, the log-normal distribution is often used to model the time to failure of systems or components, particularly when the data are skewed. Various types of censoring and truncation are considered in lifetime data analysis. Right-censored and left-truncated right-censored (LTRC)  data are common in lifetime analysis.  One can develop JEL-based tests under these situations. The U-statistics for right-censored and LRTC data (Datta et al. (2010) and Sudheesh et al. (2023)) can be considered while developing the JEL ratio test in these scenarios.


\begin{thebibliography}{999}
\bibitem{} Ahmadi, J. and Balakrishnan, N. (2024). Characterizations of continuous log-symmetric distributions based on properties of order statistics. {\em Statistics}, 1--25.
\bibitem{} Batsidis, A., Economou, P., and Tzavelas, G. (2016). Tests of fit for a lognormal distribution. {\em Journal of Statistical Computation and Simulation}, 86, 215--235.

\bibitem{} Datta, S., Bandyopadhyay, D. and  Satten, G. A. (2010). Inverse probability of censoring weighted U‐statistics for right‐censored data with an application to testing hypotheses. {\em Scandinavian Journal of Statistics}, 37, 680--700.

\bibitem{} Deepesh B., Sudheesh, K. K. and Sreelakshmi N (2021). Jackknife Empirical likelihood inference for probability weighted moments. {\em Journal of the Korean Statistical Society}, 50, 98--116.
\bibitem{} Deemat C Mathew., Reeba Mary Alex and Sudheesh K. K. (2024). Jackknife empirical likelihood ratio test for testing mean time to failure order. {\em Statistical Papers}, 65, 79--92.

\bibitem{} Ferrari, S. L., and Fumes, G. (2017). Box–Cox symmetric distributions and applications to nutritional data. {\em AStA Advances in Statistical Analysis}, 101, 321--344.


\bibitem{Jing 09}Jing, B. Y., Yuan, J. and Zhou, W. (2009). Jackknife empirical likelihood. {\em Journal of the American Statistical Association}, 104, 1224--1232.

\bibitem{} Jiang, H. and  Zhao, Y. (2024). Bayesian empirical likelihood inference for the mean absolute deviation. {\em Statistics}, 58, 277--301.

\bibitem{}
Owen, A. B. (1988). Empirical likelihood ratio confidence intervals for a single functional. {\em Biometrika}, 75, 237--249.
\vskip4pt
\bibitem{} Owen, A. (1990). Empirical likelihood ratio confidence regions. {\em The Annals of Statistics}, 18, 90--120.

\bibitem{} Marshall, A.W. and Olkin, I. (2007). {\em Life Distributions; Structure of
Nonparametric, Semiparametric, and Parametric Families}. Wiley, Hoboken.

\bibitem{}Seshadri, V. (1965). On random variables which have the same distribution as their reciprocals. {\em Canadian Mathematical Bulletin}, 8, 819--824.


\bibitem{} Saparya, S. and Sudheesh K. K. (2024). Jackknife empirical likelihood ratio test for testing the equality of semivariance. {\em Statistical Papers}, Accepted.

\bibitem{} Sathar, E.I.A. and Jose, J. (2023) Extropy based on records for random variables representing residual life. {\em Communications in Statistics-Simulation and Computation}
52, 196--206.

\bibitem{} Sudheesh K. K., Isha Dewan and Litty Mathew. (2022). Jackknife empirical likelihood ratio test for testing mean residual life and mean past life ordering. {\em Statistics}, 56, 1012--1028.

\bibitem{}  Sudheesh, K. K., Anjana, S. and Xie, M. (2023). U-statistics for left truncated and right censored data. {\em Statistics}, 57, 900--917.

\bibitem{ }Sudheesh, K. K., Sreedevi, E.P. and Balakrishnan, N. (2024). Relationships between cumulative entropy/extropy, Gini mean difference and probability weighted moments.  {\em Probability in the Engineering and Informational Sciences}, 38, 28--38.
\bibitem{TG75}Thomas, D. R. and  Grunkemeier, G. L. (1975). Confidence interval estimation of survival probabilities for censored data. {\em Journal of the American Statistical Association}, 70, 865--871.

\bibitem{} Vanegas, L. H. and Paula, G. A. (2016).  Log-symmetric distributions: statistical properties and parameter estimation. {\em Brazilian Journal of Probability and Statistics}, 30, 196--220.
\bibitem{} Vexler, A., Zou, L. and Hutson, A. D. (2017). An extension to empirical likelihood for evaluating probability weighted moments. {\em Journal of Statistical Planning and Inference},  182,  50--60.



\end{thebibliography}
\end{document}